%% file: luria-delbruck.tex
\documentclass[12pt]{article}
\usepackage{amsmath}
\usepackage{amsthm}
\usepackage{psfrag}
\usepackage{graphicx}
\usepackage{url}
\usepackage{tikz}
\usepackage[LGR,T1]{fontenc}
\usepackage{setspace}
\usepackage{fullpage}

\newcommand{\textgreek}[1]{\begingroup\fontencoding{LGR}\selectfont#1\endgroup}

\newtheorem{theorem}{Theorem}[section]
\newtheorem{proposition}{Proposition}[section]
\numberwithin{equation}{section}

\title{Estimation of Mutation Rates from Fluctuation Experiments via Probability Generating Functions\thanks{Research supported by NSF DMS Grant 0928053, PRISM: Mathematics in Life Sciences}}

\author{Stephen Montgomery-Smith\thanks{Department of Mathematics, University of Missouri, Columbia MO 65211},
Anh Le\thanks{Student in Mathematics of Life Sciences Program, University of Missouri, Columbia MO 65211},
George Smith\thanks{Department of Biology, University of Missouri, Columbia MO 65211}
\\
Sidney Billstein\footnotemark[3],
Hesam Oveys\footnotemark[2],
Dylan Pisechko\footnotemark[3],
Austin Yates\footnotemark[3]}
\date{}

\begin{document}

\maketitle

\begin{abstract}
This paper calculates probability distributions modeling the Luria-Delbr\"uck experiment.  We show that by thinking purely in terms of generating functions, and using a `backwards in time' paradigm, that formulas describing various situations can be easily obtained.  This includes a generating function for Haldane's probability distribution due to Ycart.  We apply our formulas to both simulated and real data created by looking at yeast cells acquiring an immunization to the antibiotic canavanine.

This paper is somewhat incomplete, having been last significantly modified in March 29, 2014.  However the first author feels that this paper has some worthwhile ideas, and so is going to make this paper publicly available.
\end{abstract}

\section{Introduction}

The famous experiment of Luria and Delbr\"uck \cite{luria-delbruck}, \cite{g-smith-et-al} determined whether mutations in bacteria arise via Darwinian evolution or some kind of Lamarckian process.  The experiment consisted of taking small samples of bacteria, and then allowing them to replicate up to a known large number of cells on many plates, and then looking at the distribution of number of cells on each plate that had acquired an immunity to an antibiotic.  The mutation was carefully chosen to be a forward mutation, that is, the probability of the mutation taking place was very much larger than the probability of the mutation reversing itself.

If the immunity was acquired by some kind of Lamarckian process, then one would naturally expect the distribution of surviving cells in the plates to follow a Poisson distribution.  In particular, the variance would have a value very close to the mean, and the probability of a plate having a so called `jackpot,' that is, a much larger number of surviving cells than the mean, would be vanishingly small.

However, Luria and Delbr\"uck argued that if the immunity was acquired by Darwinian evolution, then the variance of the number of surviving cells would be much larger than the mean, and furthermore, the probability of any plate having a `jackpot' would be large enough that it would be observed reasonably often.

Their experiment conclusively validated the latter assumption.  Soon after this, scientists sought after a formula for the distribution of surviving cells in each plate.  According to \cite{sarkar}, the biologist J.B.S.~Haldane produced a formula for the distribution in 1946.  But his work was never widely published at that time, and in 1949 Lea and Coulson \cite{lea-coulson} produced a different distribution, which seems to be the basis of most research these days.

A great advantage of their distribution is that it can be given by a simple generating function, whereas the calculations described by Haldane seem to be very time consuming, and requires enumerating combinatorial structures.  For example, the paper \cite{ma} gives an algorithm whereby the Lea-Coulson distribution can be calculated quickly.  The survey paper \cite{rosche} gives many ways in which one can calculate the \emph{fluctuation rate}, that is, the probability of a single offspring cell acquiring the mutation.  And the paper \cite{asteris} gives a Bayesian approach to estimating the fluctuation rate.

The primary goal of this paper is to promote an approach which uses generating functions from the very beginning.  The other change we propose to the approaches taken in other papers is to think `backwards in time' rather than `forwards in time,' that is, instead of considering how mutants might have developed from initial conditions, look at the final situation and reason out where the mutations must have come from.  We will illustrate this approach under several conditions.  In particular we will be able to obtain a generating function for the Haldane distribution which allows for rapid calculations.

\section{The experimental setup to be modeled}

We assume that in each plate we start from a small number of cells, none of which have the mutation.  Then the cells replicate repeatedly.  When each cell replicates, one of the offspring can acquire the mutation with a probability $\mu$, which is extremely small, effectively infinitesimal.  We continue the replication process until we have obtained a population of $n$ cells, where $n$ is effectively infinite.  We assume that $m=\mu n$ is of order 1.

We shall use the following terminology (which is common in the literature).  We shall call a cell that acquires the mutation, but not from its parent, a \emph{mutation}.  A cell that has the mutant gene will be called a \emph{mutant}.  We will assume that the probability of a mutant have offspring that are not mutants is zero.  This can either come from supposing that the genetic change is much more likely to happen than to be reversed (that is, it is a forward mutation), or by realizing that the population of mutants is so much smaller than the number of non-mutants, that even if the probability of a mutant having non-mutant offspring were of the order $\mu$, we wouldn't observe this anyway.)  Thus all mutations are mutants, but not vice versa.

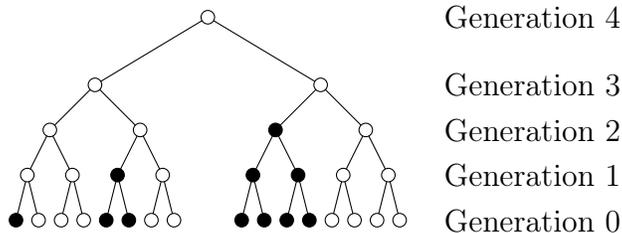
\begin{figure}
\begin{center}
\input{gen}
\caption{Five generations producing 7 mutants from 3 mutations.}
\label{generations}
\end{center}
\end{figure}

The quantity $m$ is the fundamental parameter that describes the probability distribution of the random variable $R$ which is the number of cells that have acquired the mutation.  All the formulas described in this paper will be of the form
\begin{equation}
\label{of the form}
\Pr(R = r) = p_r(m) = q_r(m) e^{-\alpha m}
\end{equation}
where $m=\mu n$ is the fundamental parameter that we try to estimate, $\alpha$ is a real number, and $q_r(m)$ is a polynomial of degree $r$.

We assume that all cells take approximately the same amount of time between cell divisions.  We use the letter $g$ to denote the number of generations backwards in time, setting $g = 0$ to denote when the experiment finishes.  Under this assumption, the cell divisions that took place recently (that is, with $g$ small), are assumed to take place synchronously.  However the bulk population may or may not be assumed to be dividing asynchronously.  We give each mutant a \emph{generation number} $g\ge 1$, which says this mutant was created $g$ generations ago.  In future papers we will drop the assumption of the same amount of time between cell divisions.

Another assumption we make is that the population of non-mutant cells is effectively infinite.  Thus the number of non-mutant cells is effectively the same as the number of all cells $g$ generations ago, and does not depend upon the number of mutations created.

Finally, we will allow for the possibility that cells may die or become non-functional.  If cell death does not occur, then the total number of cells $g$ generations ago is $2^{-g} n$.  Figure~\ref{generations} illustrates the process looking only four generations back, but in our formulas we assume that the total number of generations is effectively infinite.

Haldane produced a distribution for the Luria-Delbr\"uck experiment based on the assumption that all the cells divide with perfect synchronicity, and no cells die.  The formula can be described thus: given an integer $r\ge 0$, let $P_r$ denote the set of sequences $(a_0,a_1,\dots)$ of non-negative integers, with only finitely many non-zero terms, such that $\sum_{s=0}^\infty a_s 2^s = r$.  Then
\begin{equation}
\label{haldane prob}
\Pr(R = r) = e^{-m} \sum_{(a_s)\in P_r} (m/2)^{\sum_{s=0}^\infty a_s} \Big/ \prod_{s=0}^\infty a_s\!! \, 2^{s a_s}
\end{equation}
(Note that \cite{sarkar} has a typographical error in stating this formula.  Note also that his $g$ is our $m/2$.)

\section{Generating functions for probability distributions}

If $R$ is a random variable that takes values in the non-negative integers, then the \emph{probability generating function} of $R$ is defined by the formula
\begin{equation}
G_R(x) = \sum_{r=0}^\infty \Pr(R = r) x^r
\end{equation}
Note that $G_R$ is an analytic function, with radius of convergence at least as big as 1.  We will make much use of the following well known results.

\begin{proposition}\label{gen}
Let $X_k$ ($1 \le k < \infty$) be a sequence of non-negative integer valued independent random variables.
\begin{enumerate}
\item If the sum $Z = \sum_{k=1}^\infty X_k$ converges almost surely, then
\begin{equation}
G_Z(x) = \prod_{k=1}^\infty G_{X_k}(x)
\end{equation}
\item If $X_k$ are identically distributed, and $N$ is a non-negative integer valued random variable that is independent of the $X_k$, and if $Z = \sum_{n=1}^N X_k$, then
\begin{equation}
G_Z(x) = G_N(G_{X_k}(x))
\end{equation}
\item If $N$ is a Poisson random variable with mean $\lambda$, then
\begin{equation}
G_N(x) = e^{\lambda (x-1)}
\end{equation}
\end{enumerate}
\end{proposition}
All the distributions we shall consider will have a generating function of the form
\begin{equation}
\label{G H}
G_R(x)  = e^{-\alpha m} e^{m H_R(x)}
\end{equation}
where $\alpha > 0$, and $H_R(x)$ is an analytic function satisfying $H_R(0) = 0$ with Taylor series
\begin{equation}
H_R(x) = \sum_{k=1}^\infty h_k x^k
\end{equation}
with $h_k \ge 0$ for all $k \ge 1$.

For example if the mutation acquisition were Lamarckian, then $R$ would have the Poisson distribution with mean $m$, that is
\begin{equation}
p_r(m) = e^{-m} m^r/r! , \quad G_R = e^{-m} e^{m x}
\end{equation}
and it is well known that the generating function is

Another example is the Lea-Coulson distribution \cite{lea-coulson}.  This assumes that the mutation acquisition is Darwinian, but it also approximates the discrete replication process by a continuous process.  They obtain a distribution that satisfies~\eqref{G H} with $\alpha = 1$ and
\begin{equation}
\label{lea-coulson}
H_R(x) = \sum_{k=1}^\infty \frac{x^k}{k(k+1)}
=
\frac{x + (1-x)\log(1-x)} x
\end{equation}

There are two ways to compute the polynomials $q_r$ from the power series for $H_R(x)$.  One way is to place $H_R(x)$ into the Taylor's series for $e^x$.  This involves multiplication of polynomials.  We discuss this method in Section~\ref{section fourier}.

The other algorithm is described in Ma, Sandri, and Sarkar \cite{ma}:
\begin{gather}
q_0(m) = 1 \\
\label{ma}
q_{r}(m) = m \sum_{s=1}^{r} \frac sr \, h_{s} q_{r-s}(m)
\end{gather}
In practice, this algorithm worked very well.  Note that to compute $q_r(m)$, one only needs to know $h_s$ for $s \le r$.

\section{The generating function approach}

The approach adopted throughout this paper is to compute the generating function for $R$, the number of mutations.  Let us first illustrate this method to derive the Haldane distribution.  We assume the cell division is completely synchronous, and that no cells die or malfunction.

While this paper was being prepared, we discovered that the following result is a special case of a result that appeared in 2013 \cite{ycart}.

\begin{theorem} \label{haldane thm}
If $R$ has the Haldane distribution, then its generating function is given by equation~\eqref{G H} with $\alpha = 1$ and
\begin{equation}
\label{haldane equ}
H_R(x) = \sum_{g=0}^\infty \frac{x^{2^{g}}}{2^{g+1}}
\end{equation}
Hence from the Ma-Sandri-Sarkar algorithm equation~\eqref{ma} we obtain the rapid formula
\begin{equation}
\label{ma haldane}
q_{r}(m) = \frac m{2r} \sum_{g=0}^{\lfloor\log_2 r\rfloor} q_{r-2^{g}}(m)
\end{equation}
\end{theorem}
From equation~\eqref{ma haldane} we obtain the polynomials given in Table~\ref{fig-probs}.

\begin{table}
\begin{tabular}{|c|c|}
\hline
$r$ & $\Pr(X = r)$ \\
\hline
 $0$ & $e^{-m}$ \\
 $1$ & $\frac{m}{2} \, e^{-m}$ \\
 $2$ & $\left(\frac{m^2}{8}+\frac{m}{4}\right) e^{-m}$ \\
 $3$ & $\left(\frac{m^3}{48}+\frac{m^2}{8}\right) e^{-m}$ \\
 $4$ & 
   $\left(\frac{m^4}{384}+\frac{m^3}{32}+\frac{m^2}{32}+\frac{m}{
   8}\right) e^{-m}$ \\
 $5$ & 
   $\left(\frac{m^5}{3840}+\frac{m^4}{192}+\frac{m^3}{64}+\frac{m
   ^2}{16}\right) e^{-m}$ \\
 $6$ & 
   $\left(\frac{m^6}{46080}+\frac{m^5}{1536}+\frac{m^4}{256}+\frac{7 m^3}{384}+\frac{m^2}{32}\right) e^{-m}$ \\
 $7$ & 
   $\left(\frac{m^7}{645120}+\frac{m^6}{15360}+\frac{m^5}{1536}+\frac{m^4}{256}+\frac{m^3}{64}\right) e^{-m}$ \\
 $8$ & 
   $\left(\frac{m^8}{10321920}+\frac{m^7}{184320}+\frac{m^6}{1228
   8}+\frac{m^5}{1536}+\frac{25
   m^4}{6144}+\frac{m^3}{256}+\frac{m^2}{128}+\frac{m}{16}\right
   ) e^{-m}$ \\
 $9$ & 
   $\left(\frac{m^9}{185794560}+\frac{m^8}{2580480}+\frac{m^7}{12
   2880}+\frac{m^6}{11520}+\frac{3
   m^5}{4096}+\frac{m^4}{512}+\frac{m^3}{256}+\frac{m^2}{32}\right) e^{-m}$ \\
 $10$ & 
   $\left(\frac{m^{10}}{3715891200}+\frac{m^9}{41287680}+\frac{m^
   8}{1474560}+\frac{7 m^7}{737280}+\frac{5 m^6}{49152}+\frac{61
   m^5}{122880}+\frac{m^4}{768}+\frac{5
   m^3}{512}+\frac{m^2}{64}\right) e^{-m} $\\
\hline
\end{tabular}

\

\caption{Probabilities according to Haldane's distribution.}
\label{fig-probs}
\end{table}

\begin{proof} We know that $g$ generations ago, the population count is $n / 2^g$.  All of these were created from cells dividing $g+1$ generations ago, and therefore half of these have copied genotypes, the rest having the original genotype.  Thus the number of mutants created at that time is $N_g$, a Poisson random variable with parameter $\frac12\mu n / 2^g = 2^{-g-1} m$.  The number of mutations that arise from the $g$th generation mutants is $2^g N_g$.  Hence
\begin{equation}
R = \sum_{g=0}^\infty 2^g N_g
\end{equation}
Furthermore, because we suppose the total population of cells to be very large, we can assume that the random variables $N_g$ are independent.  Thus we obtain the generating function
\begin{equation}
\label{gen fun Haldane}
\begin{aligned}
G_R(x) &= \prod_{g=0}^\infty G_{N_g}(x^{2^g}) = \prod_{g=0}^\infty e^{2^{-g-1} m (x^{2^g}-1)} \\
&= \exp\left(m \sum_{g=0}^\infty 2^{-g-1} (x^{2^g}-1)\right)
= e^{-m} \exp\left(m \sum_{g=0}^\infty 2^{-g-1} x^{2^g}\right)
\end{aligned}
\end{equation}
\end{proof}

This approach can be generalized as follows.  First, we define times $t_k$, ($k = 0, 1, 2, \dots$), with $0 = t_0 < t_1 < t_2 < \dots$.  The units of time will be generations, and as with generation number, a positive number denotes the number of generations before the experiment is completed.  Let us suppose that a time $t_k$, ($k \ge 1$), and only at those times, some or all of the cells divide.  For each $k$, we let $n_k$ denote the total population of cells at time $t_k$, and let $\pi_k = n_k/n$.  Thus $1 = \pi_0 \ge \pi_1 \ge \pi_2 \ge \dots > 0$.  Because we have assumed that the total population of cells is effectively infinite, we can assume $t_k \to \infty$ as $k \to \infty$.

We consider two random variables:
\begin{enumerate}
\item for each possible $k \ge 0$, the number of mutants $N_k$ created at negative time $t_k$, which will always be a Poisson random variable with parameter $\pi_{k} - \pi_{k+1}$;
\item for each possible $k \ge 0$, the number of mutations $M_k$ that arise from any one mutant which was created at time $t_k$.
\end{enumerate}
Then the total number of mutations will always be given be the formula
\begin{equation}
R = \sum_{k=0}^\infty \sum_{l=1}^{N_k} M^{(l)}_k
\end{equation}
where $M^{(l)}_g$ denotes independent copies of $M_l$.  Hence we obtain the generating function
\begin{equation}
\label{gen fun approach}
\begin{aligned}
G_R(x) &= \prod_{k=0}^\infty G_{N_k}(G_{M_k}(x)) = \prod_{k=1}^\infty e^{m (\pi_{k} - \pi_{k+1}) (G_{M_k}(x)-1)} \\
&= \exp\left(m \sum_{k=0}^\infty (\pi_{k} - \pi_{k+1}) (G_{M_k}(x)-1)\right)
\end{aligned}
\end{equation}
Let's explain first how the Haldane distribution fits into this paradigm.  Since cell division is synchronous, we set $t_g = g$ where the time is measured in multiples the time between cell divisions.  Since no cells die, we have that $\pi_g = 2^{-g}$, and $M_g = 2^{g}$, so that $G_{M_g}(x) = x^{2^{g}}$.  Substituting into equation~\eqref{gen fun approach}, we obtain equation~\eqref{haldane equ}.

We can also obtain the Lea-Coulson distribution using this paradigm.  Note the argument we present is in their paper as a ``second proof" (and indeed the whole ``generating function approach" is inspired by their paper).  We still suppose that no cells die, but we assume that the formulas $\pi_k = 2^{-t_k}$ and $M_k = 2^{t_k}$ still hold when $t_k$ is not an integer.  To get the most even spacing of the generating of cells, we set $t_k = \log_2(k+1)$.  Then $\pi_k - \pi_{k+1} = 1/(k+1) - 1/(k+2) = 1/(k+1)(k+2)$, and $M_k = k+1$.  Thus equation~\eqref{lea-coulson} follows after substituting the summation variable $k$ by $k-1$.

\section{Haldane's distribution versus the Lea-Coulson distribution}

The Lea-Coulson assumes that the generations are somehow being created continuously.  From this point of view, it would seem that their distribution is unrealistic.

However there also are problems with the Haldane distribution.  It is obvious that the probability of obtaining $128$ mutants is far higher than the probability of obtaining $127$ mutants.  If one had a single mutation, say, seven generations ago, this will result in $128$ mutants.  But it is possible that a few of these mutants will either die (and perhaps any of their mutant ancestors died), or fail to be transferred to the plate properly.  Thus while Haldane's distribution puts a very low probability of getting $127$ mutants compared to $128$ mutants, nevertheless from an experimental point of view, it is very likely that not all of the $128$ mutants will be observed, and the chances of seeing $127$ or less mutants is comparable with the probability of seeing $128$ mutants.

To show this effect, in Figure~\ref{log log} we give log/log plots of $p_r(m)$ to $r$, for the value $m=2$, for both distributions.  On the $x$-axis we show $\log_2(r)$, and on the $y$-axis we show $\log_{10}(p_r(2))$.  Notice that for the Lea-Coulson distribution that this graph approaches a straight line as $r \to \infty$.  In fact, in \cite{ma}, it is shown that $p_r(m) \approx c_m r^{-2}$ as $r\to\infty$, where $c_m$ depends only on $m$.

\begin{figure}
\begin{center}
\includegraphics[scale=0.45]{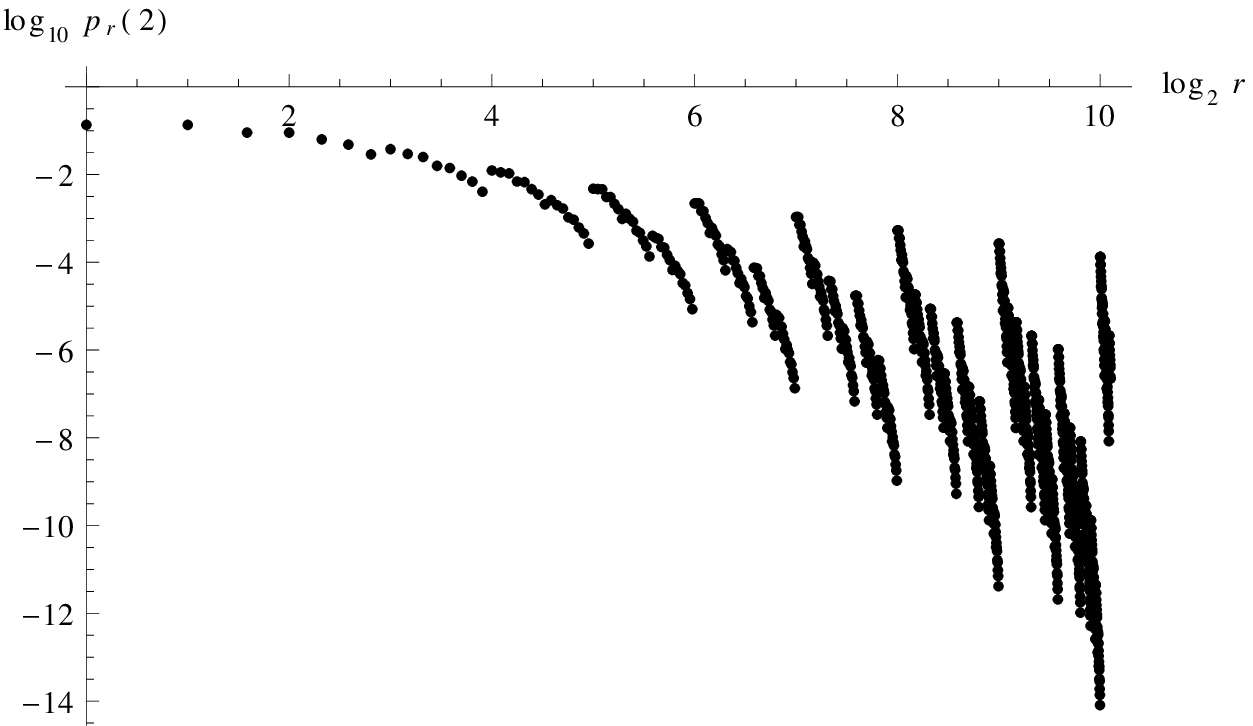}
\includegraphics[scale=0.45]{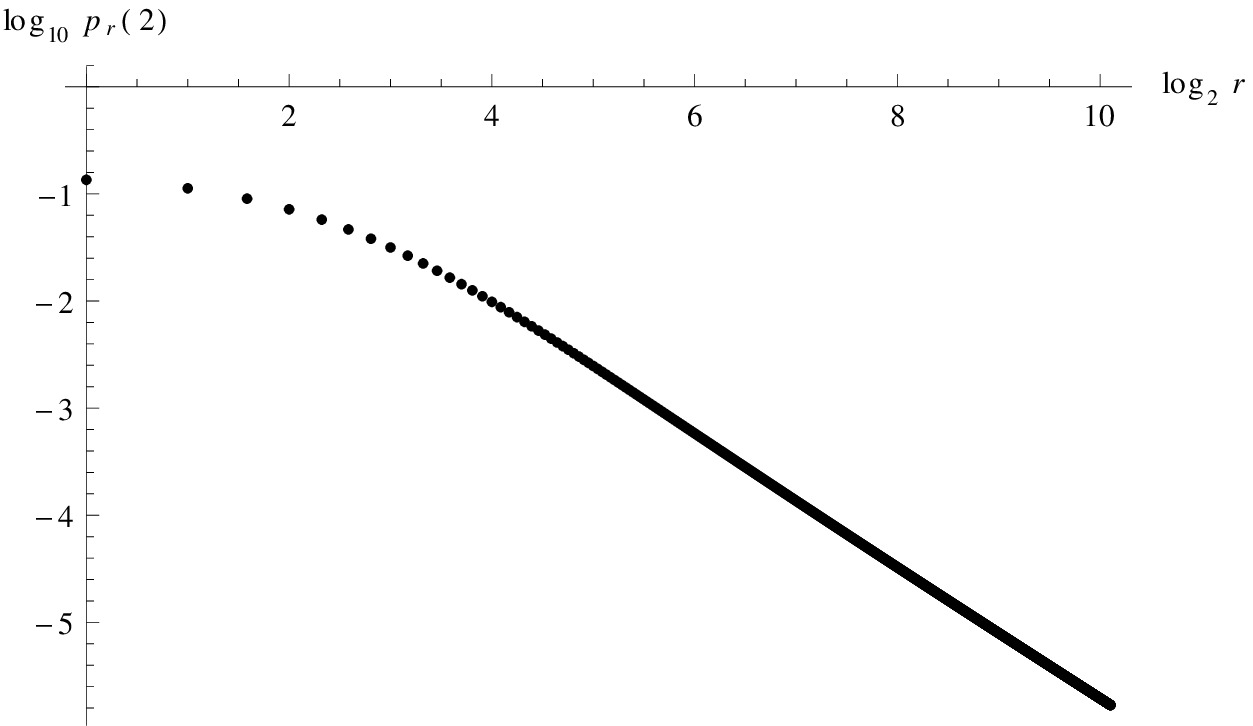}
\caption{Log/Log graphs of $p_r(2)$ against $r$, for Haldane's distribution, and for the Lea-Coulson distribution.}
\label{log log}
\end{center}
\end{figure}

We will look at several ways to modify the Haldane distribution.
\begin{enumerate}
\item Allow that some cells may die.
\item Examine the situation when only a fixed proportion of cells are plated with the antibiotic.
\item Consider asynchronous cell division.
\item Allow that the mutants replicate at a different rate than the non-mutants \cite{koch}.
\end{enumerate}
We will explore the third case only superficially, as we have not yet obtained formulas that fully handle this issue.

\section{Accounting for cell death}

Let us now assume that a cell dies with probability $\theta$ before it replicates.  Thus from one cell, the expected number of cells that come from this cell one generation later is $2(1-\theta)$.  Thus
\begin{equation}
\pi_g = (2(1-\theta))^{-g}
\end{equation}
The computation of $G_{M_g}$ is given by a recursion relationship.  From a mutant created $g$ generations ago, with probability $\theta$ this mutant will die, and with probability $1-\theta$ there will by $X+Y$ cells, where $X$ and $Y$ are independent and have the same distribution as $M_{g-1}$.  Thus we obtain
\begin{gather}
G_{M_0}(x) = x \\
\label{recurrence}
G_{M_{g+1}}(x) = \theta + (1-\theta) [G_{M_g}(x)]^2
\end{gather}
These formulas are then substituted into equation~\eqref{gen fun approach} to obtain equation equation~\eqref{of the form} with
\begin{gather}
\alpha = \sum_{g=0}^\infty (1-2\theta)(2(1-\theta))^{-g-1} (G_{M_g}(0) - 1) \\
H_R(x) = \sum_{g=0}^\infty (1-2\theta)(2(1-\theta))^{-g-1} (G_{M_g}(x) - G_{M_g}(0))
\end{gather}
It is unlikely we will find an explicit formula for solving the recurrence relation~\eqref{recurrence}, because this is in essence the same recurrence relation that is used to define the famous Mandelbrot set \cite{mandelbrot}.

\section{Plating only a fixed proportion of the cells}

This is discussed for the Lea-Coulson distribution in \cite{stewart:a,stewart:b}.

Let us now assume that we perform an experiment such that the probability of obtaining $R = r$ mutants is given by $p_r$, with corresponding generating function $F_R(x) = e^{-\alpha m} \exp(m H_R(x))$.  Now suppose that we only plate a proportion $\theta \in (0,1]$ of the cells.  How many mutants will we observe?  Let us call the number of mutants observed $R_\theta$.

\begin{theorem} The random variable $R_\theta$ has generating function
\begin{equation}
G_\theta(x) = G_R(1-\theta + \theta x) = e^{-\alpha_\theta m} \exp(m H_\theta(x))
\end{equation}
where
\begin{gather}
\alpha_\theta = \alpha - H_R(1-\theta) \\
H_\theta(x) = H_R(1-\theta + \theta x) - H_R(1-\theta)
\end{gather}
\end{theorem}

\begin{proof}
We have that $R_\theta$ has the same distribution as $\sum_{r=1}^R I_r$, where $I_r$ is a sequence of independent random variables taking the value $1$ with probability $\theta$, and the value $0$ with probability $1-\theta$.  The result follows from applying Proposition~\ref{gen}.

Another way to obtain this formula is to use the binomial distribution
\begin{equation}
\Pr(R_\theta = r) = \sum_{n=r}^\infty \Pr(R=n) \Pr(R_\theta = r \big| R=n) 
= \sum_{n=r}^\infty \binom nr \theta^r(1-\theta)^{n-r} p_n
\end{equation}
and then rearrange the resulting double sum.
\end{proof}

Usually, the most timely and accurate numerical method for computing $\alpha_\theta$, and coefficients of the power series for $H_\theta$, is by direct evaluation from the power series of the derivatives of $H(x)$ at $x=1-\theta$.  For the Lea-Coulson probability measure, it makes sense to compute $\alpha$ directly by substituting equation~\eqref{lea-coulson}.  One might also be tempted to compute the coefficients of $H_\theta$ by symbolically differentiating $H_R$ given in equation~\eqref{lea-coulson}.  However the symbolic derivatives become very unwieldy.

\begin{figure}
\begin{center}
\includegraphics[scale=0.45]{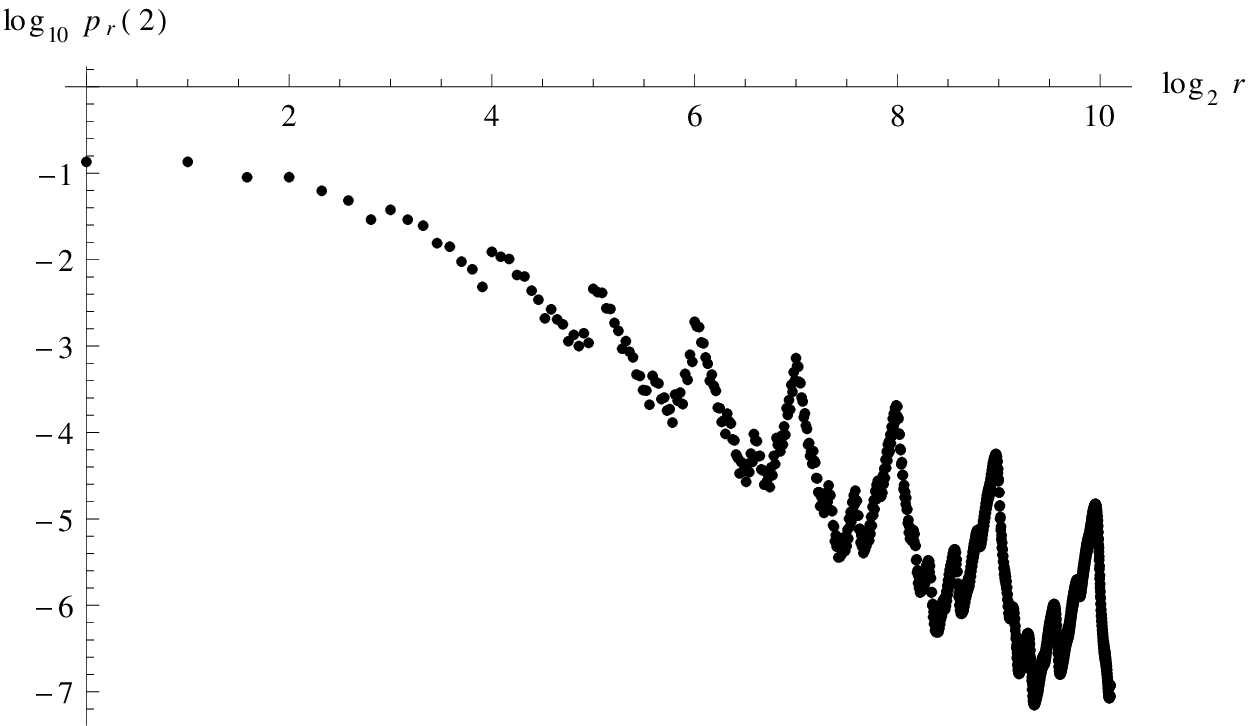}
\includegraphics[scale=0.45]{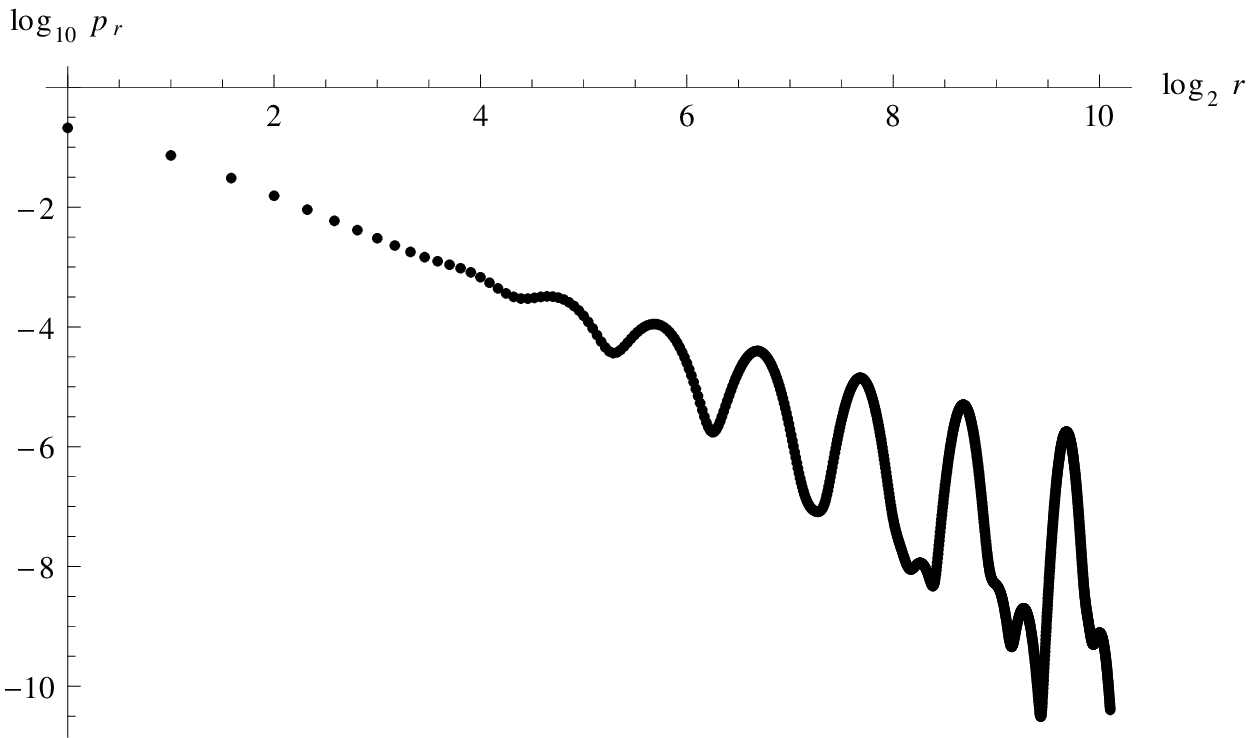}
\caption{Log/Log graph of $p_r = \Pr(R_\theta=r)$ against $r$, when cells die with probabilities $\theta=0.01$ (left hand side) or $\theta=0.1$ (right hand side), and $R$ has the Haldane distribution with $m = 2$.}
\label{log log 3}
\end{center}
\end{figure}

\section{Asynchronous cell division}

There are two notions of asynchronous cell division that we shall consider.
\begin{enumerate}
\item \emph{slight asynchronous cell division}  We suppose that cell division is asynchronous enough so that at time $g$ generations before the experiment ends, the cell population is given by $2^{-g} n$, even if $g$ is not an integer, but not so asynchronous so that over a few generations the cell division is well approximated by synchronous cell division.
\item \emph{full asynchronous cell division}  We suppose that cell division is sufficiently asynchronous that we should base our model on a probability distribution that describes the time for a cell to divide.
\end{enumerate}
Accounting for full asynchronous cell division seems to be very difficult, and we have not yet solved this problem\footnote{Note added in August 13, 2016: this problem has been solved in the thesis of Hesam Oveys \cite{oveys}.}.  In this chapter we shall consider slight asynchronous cell division, and show that it makes no difference to the Haldane distribution.

We assume that cell splitting takes place every $g/n$ time units, where at the end we let $n\to\infty$.  Therefore in equation~\eqref{gen fun approach}, we set $t_k$ = ${k/n}$, $\pi_k = 2^{-k/n}$, and $G_{M_k} = 2^{[k/n]}$, where $[x]$ denotes the integer part of $x$.  Then writing $k/n = g + j/n$ where $g = [k/n]$, we obtain
\begin{equation}
\begin{aligned}
G_R(x)
&= \exp\left(m \sum_{g=0}^\infty \sum_{j=0}^{n-1} (2^{1/n}-1) 2^{-g-(j+1)/n} (x^{2^g}-1)\right) \\
&= \exp\left(m \sum_{g=0}^\infty 2^{-g-1} (x^{2^g}-1)\right)
\end{aligned}
\end{equation}
where the last step comes from summing the geometric series.  Thus the formula is independent of $n$.

\section{Mutants replicating at a different rate than non-mutants}

The analogous formula to the Lea-Coulson was calculated by \cite{koch}.  We will do the analog for the Haldane distribution.  While this paper was being prepared, we found out that the this result appeared in 2013 \cite{ycart}.

Let $g$ denote the number of generations back in time as counted in the time for a mutant to replicate.  Suppose that the time for a non-mutant to replicate is $\tau$ times the time for a mutant to replicate.  Then we can apply equation~\eqref{gen fun approach} with $M_g = 2^g$, and $\pi_g = 2^{\tau g}$, to obtain equation equation~\eqref{of the form} with
\begin{equation}
\alpha = -1, \quad
H_R(x) = \sum_{g=0}^\infty (1-2^{-\tau})2^{-\tau g} x^{2^g}
\end{equation}

\section{Comparison of the Ma-Sandri-Sarkar algorithm with using the Fast Fourier Transform}
\label{section fourier}

Another way to compute the probability generating function given by equation~\eqref{G H} is to expand it using Taylor's series for $e^x$, and obtain
\begin{equation}
G_R(x)  = e^{-\alpha m} \sum_{n=0}^\infty \frac1{n!} m^n [H_R(x)]^n
\end{equation}
In order to calculate $p_r(m)$, it is only necessary to sum the series to the $r$th term, and to calculated $[H_R(x)]^n$ up to $x^r$.  It is well known that $[H_R(x)]^n$ can be rapidly calculated using the Fast Fourier Transform.  If one performs a algorithm complexity analysis, one finds that the time to compute $p_r(m)$ for $m \le n$ is $O(n^3)$ using the Ma-Sandri-Sarkar algorithm, but is $O(n^2 \log n)$ when using the Fast Fourier Transform.  (The exception is when using the Ma-Sandri-Sarkar for Haldane's distribution, when the time taken is also $O(n^2 \log n)$.)

A difficulty with the Fast Fourier Transform is getting sufficient accuracy in the coefficients of $G_R(x)$.  If one calculates $[H_R(x)]^n$ using long multiplication (and this is, in effect, what the Ma-Sandri-Sarkar algorithm is doing), then because all the coefficients of $H_R(x)$ are positive, there will be no canceling large numbers to produce small numbers in computing the coefficients.  However, the Fast Fourier Transform combines the various coefficients using complex numbers, and so the accuracy of the final answer is limited by the size of the largest coefficient.

However there is also the fast polynomial package FLINT \cite{flint}, which comes standard with the Sage software package \cite{sage}.  While all our computations were performed with Sage, we found that in our situations that the Ma-Sandri-Sarkar was much faster than using the built in polynomial multiplication.

\section{Analyzing data using likelihood functions/Bayesian statistics}

In this section, we will show how to compute the likelihood function to estimate the value of $m$ from experimentally obtained data.  Suppose we have performed $n$ separate experiments, using identical but independent protocols, and obtain counts of plates $r_1, r_2, \dots, r_n$.  Then the likelihood function is
\begin{equation}
\label{like}
L(m) = \prod_{k=1}^n \Pr(R = r_k)
\end{equation}
where $R$ has the distribution that we believe most correctly represents our situation.

The Bayesian approach is to suppose that $L(m)$ represents a probability distribution for a random variable $M$, where $\Pr(M \in [m_1,m_2])$ represents the `belief' that we have that the parameter $m$ lies between $m_1$ and $m_2$.  This is related to $L(m)$ via the formula
\begin{equation}
\Pr(M \in [m,m+d m]) = C f(m) L(m) d m
\end{equation}
where here we think of $dm$ is being infinitesimal.  Here $f(m) dm$ is usually called the prior distribution, and the constant $C$ is chosen so that the integral of the right hand side is zero.  The paper \cite{asteris} describes estimates of $M$ using Bayesian statistics, and they use prior distributions like $1$m $m^{-1}$ or $m^{-2}$.

Thus the area under likelihood function represents the probability distribution of $M$ if $f(m) = 1$.  However, since $m$ is some kind of scaling factor, it makes more sense to draw the graph of $L(m)$ on a graph where the $x$-axis represents $log(m)$.  Then the area under likelihood function represents the probability distribution of $M$ if $f(m) = 1/m$.

If we treat $M$ as a random variable, then it makes sense to also compute the expected value and variance of $M$.  The computing of $C$, $E(M)$ and $\text{var}(M)$ requires computing certain integrals of the form of linear combinations of $\int_0^\infty m^a e^{-bm} \, dm$, and so can be easily be computed using a change of variables and the $\Gamma$ function.

However, since we are interested in graphing these using a log scale on the $x$ axis, to get an idea of which range of $m$ we should use for plotting, we want to compute $E(\log(M))$ and $\text{var}(\log(M))$.  These are not so easily computed, but numerical integration does an outstanding job.

In any case, we would hope that the amount of data collected should be enough so that the formulas $E(M)$ or $\exp(E(\log(M)))$ or the maximum of $L(m)$, should give similar answers, even if different prior distributions are used.  And we would hope to see a bell-shape distribution, so we should get similar answers for the estimate of the middle 68.3\%-ile of belief, that is, $E(M) \pm  \text{var}(M)$ or $\exp(E(\log(M)) \pm \text{var}(\log(M)))$.

Finally, in performing the computations, the numbers involved tend to be much smaller than much computer software can represent.  For example, the commonly used double precision IEEE floating point representation cannot represent numbers much smaller than $10^{-308}$, and values far smaller than this will often arise in calculating the likelihood function.  However we used the software package sage \cite{sage}, and this was well capable of handling very small numbers.


\section{Simulations}

\section{Analysis of real data}

The following experiments were performed by students who were part of the `Mathematics in Life Science' program at the University of Missouri.

The mutations in question confer resistance to an antibiotic (canavanine) in yeast cells.  There were two types of resistant mutants, one giving red colonies, the other giving white colonies.  Our experiment focused on the red mutants, because we were able to sequence their DNA to find the exact mutation responsible for resistance.

In Experiment~A, the students plated 70 60\textgreek{m}l cultures, each arising from a 25ml culture.  In Experiment~B was to plate 50 60\textgreek{m}l cultures from a single 25ml culture.  Each culture started at a cell density of $10^5$ cells/ml and ended at a cell density of $10^8$ cells/ml.  See Figure~\ref{experiment}.

\begin{figure}
\begin{center}
\includegraphics[scale=0.25]{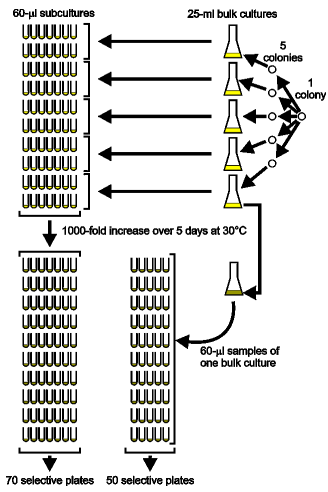}
\caption{A diagram illustrating the experiment.}
\label{experiment}
\end{center}
\end{figure}

In Experiment~A, two of the 70 cultures were accidentally destroyed.  The counts for the other 68 cultures were as follows: 2, 4, 0, 0, 0, 0, 244, 55, 0, 141, 0, 0, 0, 1, 0, 511, 0, 0, 0, 0, 4, 0, 0, 0, 0, 5, 95, 2, 0, 0, 11, 0, 0, 0, 5, 0, 0, 49, 0, 1, 0, 0, 0, 0, 0, 0, 3, 0, 0, 0, 3, 0, 0, 0, 11, 1, 16, 0, 0, 0, 1, 0, 1, 4, 1, 2, 0, 0.  These were analyzed using both the Haldane and Lea-Coulson distributions, with $\theta = 0.0024$.

mean = 0.976689,
stand dev = 0.126925.

\begin{figure}
\begin{center}
\includegraphics[scale=0.25,angle=-90]{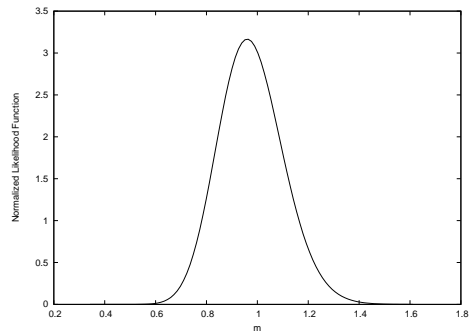}
\caption{The normalized likelihood function for real data.}
\label{post plot 2}
\end{center}
\end{figure}



The counts for Experiment~B were as follows: 155, 188, 189, 191, 173, 161, 164, 221, 191, 221, 148, 173, 186, 152, 195, 86, 90, 154, 133, 149, 165, 182, 162, 144, 129, 65, 165, 159, 151, 183, 170, 130, 140, 118, 162, 132, 154, 142, 134, 151, 150, 89, 147, 143, 142, 191, 93, 119, 163, 133.  Since these were all taken from the same culture, we only took the total value, which was 7628, and analyzed this using both the Haldane and Lea-Coulson distributions with $\theta = 0.12$.

The counts for the 50 60\textgreek{m}l cultures suggest we hit a jackpot.  However in looking at the distribution created by the Bayesian method, the data seems to rather confidently give a large figure for $m$, which we strongly suspect is too large.  One way to interpret this data is to use a log-log plot for the likelihood function.  In this case it can be seen ......  issues with how to interpret probabilities obtained using the Bayesian method.

\begin{figure}
\begin{center}
\includegraphics[scale=0.25,angle=-90]{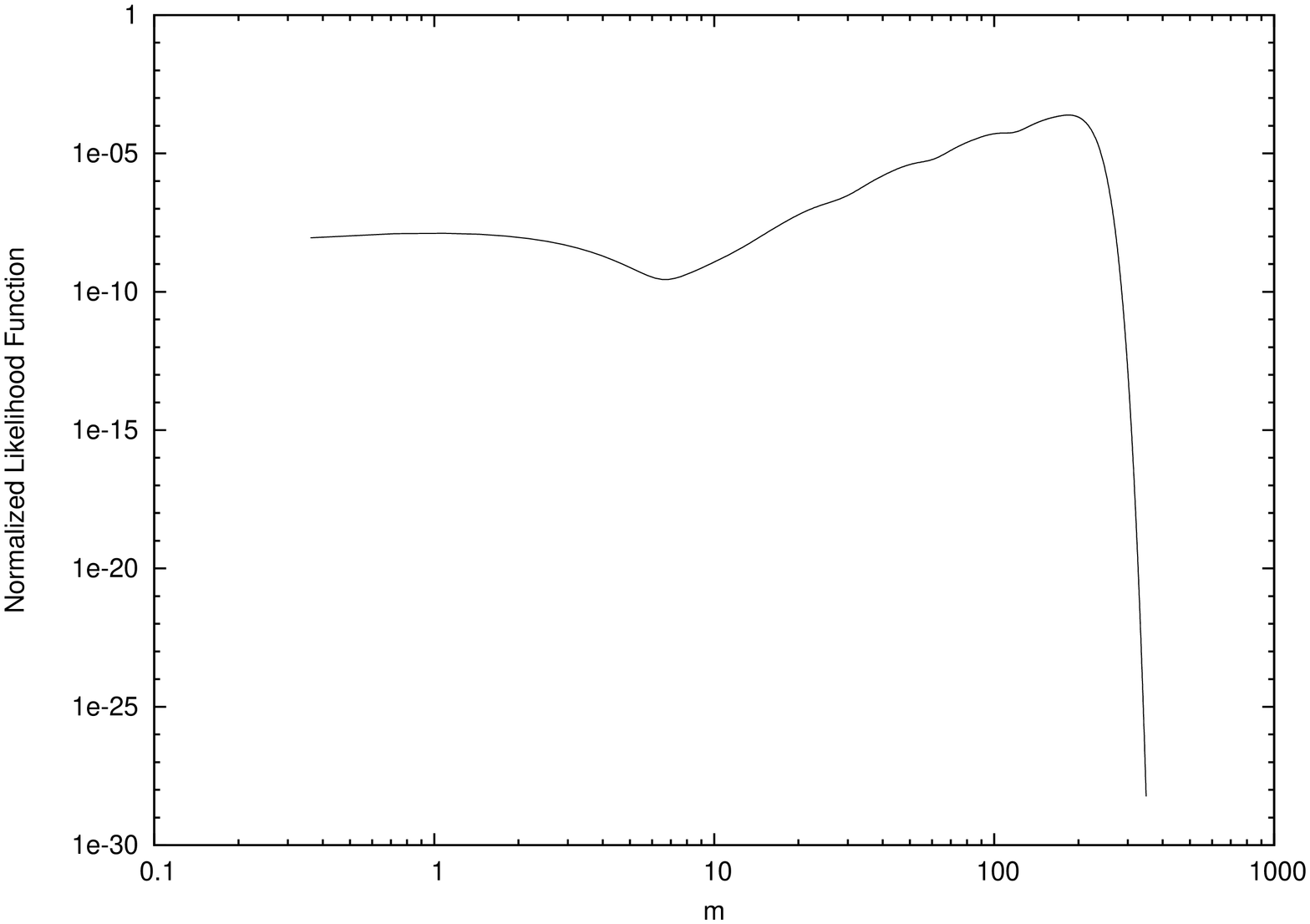}
\caption{The normalized likelihood function for the bulk culture.}
\label{post plot 3}
\end{center}
\end{figure}

\section*{Acknowledgments}

The first named author would like to thank Miriam Golomb for many useful conversations.

\end{document}

%% file: gen.tex
\begin{tikzpicture}\begin{scope}[scale=0.3]
\draw (0,0)--(5,-3);
\draw (0,0)--(-5,-3);
\draw (5,-3)--(7,-5);
\draw (5,-3)--(3,-5);
\draw (7,-5)--(8,-7);
\draw (7,-5)--(6,-7);
\draw (8,-7)--(8.5,-9);
\draw (8,-7)--(7.5,-9);
\draw (6,-7)--(6.5,-9);
\draw (6,-7)--(5.5,-9);
\draw (3,-5)--(4,-7);
\draw (3,-5)--(2,-7);
\draw (4,-7)--(4.5,-9);
\draw (4,-7)--(3.5,-9);
\draw (2,-7)--(2.5,-9);
\draw (2,-7)--(1.5,-9);
\draw (-5,-3)--(-3,-5);
\draw (-5,-3)--(-7,-5);
\draw (-3,-5)--(-2,-7);
\draw (-3,-5)--(-4,-7);
\draw (-2,-7)--(-1.5,-9);
\draw (-2,-7)--(-2.5,-9);
\draw (-4,-7)--(-3.5,-9);
\draw (-4,-7)--(-4.5,-9);
\draw (-7,-5)--(-6,-7);
\draw (-7,-5)--(-8,-7);
\draw (-6,-7)--(-5.5,-9);
\draw (-6,-7)--(-6.5,-9);
\draw (-8,-7)--(-7.5,-9);
\draw (-8,-7)--(-8.5,-9);
\fill[color=white] (0,0) circle(0.3);
\draw (0,0) circle(0.3);
\fill[color=white] (5,-3) circle(0.3);
\draw (5,-3) circle(0.3);
\fill[color=white] (7,-5) circle(0.3);
\draw (7,-5) circle(0.3);
\fill[color=white] (8,-7) circle(0.3);
\draw (8,-7) circle(0.3);
\fill[color=white] (8.5,-9) circle(0.3);
\draw (8.5,-9) circle(0.3);
\fill[color=white] (7.5,-9) circle(0.3);
\draw (7.5,-9) circle(0.3);
\fill[color=white] (6,-7) circle(0.3);
\draw (6,-7) circle(0.3);
\fill[color=white] (6.5,-9) circle(0.3);
\draw (6.5,-9) circle(0.3);
\fill[color=white] (5.5,-9) circle(0.3);
\draw (5.5,-9) circle(0.3);
\fill[color=black] (3,-5) circle(0.3);
\draw (3,-5) circle(0.3);
\fill[color=black] (4,-7) circle(0.3);
\draw (4,-7) circle(0.3);
\fill[color=black] (4.5,-9) circle(0.3);
\draw (4.5,-9) circle(0.3);
\fill[color=black] (3.5,-9) circle(0.3);
\draw (3.5,-9) circle(0.3);
\fill[color=black] (2,-7) circle(0.3);
\draw (2,-7) circle(0.3);
\fill[color=black] (2.5,-9) circle(0.3);
\draw (2.5,-9) circle(0.3);
\fill[color=black] (1.5,-9) circle(0.3);
\draw (1.5,-9) circle(0.3);
\fill[color=white] (-5,-3) circle(0.3);
\draw (-5,-3) circle(0.3);
\fill[color=white] (-3,-5) circle(0.3);
\draw (-3,-5) circle(0.3);
\fill[color=white] (-2,-7) circle(0.3);
\draw (-2,-7) circle(0.3);
\fill[color=white] (-1.5,-9) circle(0.3);
\draw (-1.5,-9) circle(0.3);
\fill[color=white] (-2.5,-9) circle(0.3);
\draw (-2.5,-9) circle(0.3);
\fill[color=black] (-4,-7) circle(0.3);
\draw (-4,-7) circle(0.3);
\fill[color=black] (-3.5,-9) circle(0.3);
\draw (-3.5,-9) circle(0.3);
\fill[color=black] (-4.5,-9) circle(0.3);
\draw (-4.5,-9) circle(0.3);
\fill[color=white] (-7,-5) circle(0.3);
\draw (-7,-5) circle(0.3);
\fill[color=white] (-6,-7) circle(0.3);
\draw (-6,-7) circle(0.3);
\fill[color=white] (-5.5,-9) circle(0.3);
\draw (-5.5,-9) circle(0.3);
\fill[color=white] (-6.5,-9) circle(0.3);
\draw (-6.5,-9) circle(0.3);
\fill[color=white] (-8,-7) circle(0.3);
\draw (-8,-7) circle(0.3);
\fill[color=white] (-7.5,-9) circle(0.3);
\draw (-7.5,-9) circle(0.3);
\fill[color=black] (-8.5,-9) circle(0.3);
\draw (-8.5,-9) circle(0.3);
\draw (10,0) node[right]{Generation $4$};
\draw (10,-3) node[right]{Generation $3$};
\draw (10,-5) node[right]{Generation $2$};
\draw (10,-7) node[right]{Generation $1$};
\draw (10,-9) node[right]{Generation $0$};
\end{scope}\end{tikzpicture}